\documentclass[twocolumn,amsmath,amssymb,superscriptaddress,prl,floatfix,reprint]{revtex4-2}
\usepackage{graphics,amssymb,amsmath,epsfig,color,textgreek}
\usepackage{amsthm}

\usepackage{graphicx}
\usepackage[dvipsnames]{xcolor}
\usepackage{dcolumn}
\usepackage{bm}
\usepackage[colorlinks=true,citecolor=cyan]{hyperref}
\hypersetup{colorlinks=true,citecolor=cyan,linkcolor=red,urlcolor=magenta}
\usepackage{braket}
\usepackage[normalem]{ulem}
\usepackage{cancel}
\usepackage{xfrac} 
\usepackage{amsmath}
\usepackage{amssymb}
\usepackage{amsthm}
\usepackage{svg}
\usepackage{lipsum}
 
\usepackage[capitalise]{cleveref}
\crefname{section}{Sec.}{Secs.}
\crefname{table}{Tab.}{Tabs.}
\crefname{figure}{Fig.}{Figs.}
\crefname{definition}{Def.}{Defs.}
\crefname{lema}{Lem.}{Lems.}
\crefname{theorem}{Thm.}{Thms.}
\crefname{corollary}{Cor.}{Cors.}

\newtheorem{theorem}{Theorem}
\newtheorem*{theorem-non}{Theorem}

\newcommand{\ham}{\mathcal{H}}

\newcommand{\Z}{\mathcal{Z}}

\begin{document}

\title{Denoising and Extension of Response Functions in the Time Domain}

\author{Alexander F. Kemper}
\email{akemper@ncsu.edu}
\affiliation{Department of Physics, North Carolina State University, Raleigh, North Carolina 27695, USA}

\author{Chao Yang}
\email{cyang@lbl.gov}
\affiliation{Computational Research Division, Lawrence Berkeley
National Laboratory, Berkeley, CA 94720, USA}

\author{Emanuel Gull}
\email{egull@umich.edu}
\affiliation{Department of Physics, University of Michigan, Ann Arbor, Michigan 48109, USA}

\date{\today}

\begin{abstract}
Response functions of quantum systems, such as electron Green's functions, magnetic, or charge susceptibilities, describe the response of a system to an external perturbation. They are the central objects of interest in field theories and quantum computing and measured directly in experiment. Response functions are intrinsically causal. In equilibrium and steady-state systems, they correspond to a positive spectral function in the frequency domain.
Since response functions define an inner product on a Hilbert space and thereby induce a positive definite function, the properties of this function can be used to reduce noise in measured data and, in equilibrium and steady state, to construct positive definite extensions for data known on finite time intervals, which are then guaranteed to correspond to positive spectra.
\end{abstract}\maketitle

\emph{Introduction.}
Response functions are critical for the understanding of physics
in a wide variety of contexts.  They
are a natural framework for considering 
 \emph{dynamical} properties
that involve excitations \cite{AGD75,Mahan,bruus,stefanucci_nonequilibrium_2013}.  Response functions are
also measured in experimental setups ranging from low-frequency
THz conductivity to magnetic susceptibilities and photoemission
spectroscopy. 
A large body of literature has been devoted to the study
of response functions, and 
a number of ``field'' theoretical approaches
avoid the calculation of eigenstates entirely, and instead cast the
formalism in terms of response functions, called
correlation or Green's functions in this context. These include embedding techniques such
as the dynamical mean field theory (DMFT) \cite{Metzner1989,Georges1992,georges1996dynamical} and its cluster extensions \cite{Hettler1998,Lichtenstein2000,Kotliar2001,maier2005quantum}, Monte Carlo approaches for lattice  \cite{Blankenbecler1982} and impurity \cite{Hirsch1986,Rubtsov2005,Werner2006,Gull2008,Gull11RMP} models, self-consistent partial summation methods for real materials and model systems \cite{Hedin1965,Bickers1989,Yeh2022}, and
non-equilibrium Green's function methods \cite{stefanucci_nonequilibrium_2013,Muhlbacher08,Werner09,Gull11BoldNoneq,Aoki2014,Schueler20,dong2022excitations,NunezFernandez22,Erpenbeck23}.

An important characteristic of correlation functions lies in their analytical properties. The retarded correlation functions have no content
at negative times due to causality\cite{AGD75,Mahan,bruus}. In the complex frequency domain, they correspond to so-called Nevanlinna functions \cite{Adamyan2009}  whose
 poles are restricted to the lower half of the complex plane.
This analytical framework has long been utilized to evaluate integrals that emerge in many-body theory \cite{AGD75,Mahan,bruus}, such as those occurring in the context of warm dense matter \cite{Vorberger12} and uniform electron liquids \cite{Filinov23}.
More recently, it was used to 
perform analytic continuation from a Wick-rotated frame to a
standard frame. Where traditional approaches that rely e.g. on the maximum
entropy method \cite{Jarrell1996} have significant uncertainty in the final result,
resulting in washed out spectra, explicitly enforcing the highly constraining analytical properties of Nevanlinna functions
results in a drastic reduction of the uncertainty, leading to sharp spectral functions \cite{fei2021nevanlinna,Fei2021Caratheodory}.
It is clear from these examples that encoding this mathematical structure into numerical algorithms can be used
to great benefit.

In this Letter, we analyze a fundamental property of correlation functions in the time domain: several correlation functions of interest are {\emph{positive definite functions}} of their time arguments.  This property arises from viewing correlation functions
as an inner product in the vector space of operators, combined with the fact
that the time translation operator is a unitary representation of the time translation
group.
This positive definiteness sets
a strong constraint on the correlation function, similar to its other analytic properties --- in fact, some of these directly follow from the positive definiteness \cite{Bochner1932}. 
Moreover,
insisting that a correlation function is positive definite
enables both the extension of numerical correlation functions
to later times, and the extraction of clean spectra from noisy data such as that obtained from
Monte Carlo and Quantum Computing approaches.

We consider a quantum system in a quantum state $\rho$.  In a second quantized formalism, quantum states can be described as states in Fock space \cite{AGD75,Mahan,bruus}. The energetics (and dynamics) of the system is described by a Hamiltonian $\ham$; in the Heisenberg picture, the time-evolution of an operator $A$ is given by $A(t)=e^{i\ham t}Ae^{-i\ham t}$.
Expectation values of operators for the state $\rho$ are computed as $\langle A(t)\rangle=\text{Tr}[\rho A(t)]$. In a system with time-translational invariance, $\ham$ commutes with $\rho$, and in the canonical ensemble $\rho=e^{-\beta \ham}/\Z$, with $\beta$ denoting the inverse temperature and
$\Z \equiv \mathrm{Tr} e^{-\beta \ham}$.

We are primarily interested in time-dependent single-particle correlation functions of the type
\begin{align}
    G_{AB}(t,t')=\langle A^\dagger(t) B(t')\rangle.
\end{align}
In a fermion system with creation operators $c_i^\dagger$ creating particles in state $i$, the so-called ``lesser'' Green's function \cite{stefanucci_nonequilibrium_2013}
$G_{ij}^{<}(t,t')=i\text{Tr}\left[\rho c_j^\dagger(t')c_i(t)\right]$ and the ``greater'' Green's function
$G_{ij}^{>}(t,t')=-i\text{Tr}\left[\rho c_i(t)c_j^\dagger(t')\right]$, as well as the charge and spin correlation functions are related to correlation functions of this type. Below, we consider these functions without their usual $\pm i$ prefactors.


\emph{Mathematical exposition.}
In this section, we will show that $G_{AB}(t,t')$ are
\emph{positive definite} functions when $A=B$.
First, we observe that the correlation function
can be viewed as an inner product of the linear operators that act on Fock space $V$ (we denote this vector
space as $\mathcal{L}(V)$). This fact was also noted in Ref.~\cite{Hyrkas22}, where it was used to construct strictly positive perturbative approximations. We define an inner product on $\mathcal{L}(V)$ as
\begin{align}
    \langle A, B \rangle := \mathrm{Tr} \left[ \rho A^\dagger B \right],
    \label{eq:innerproduct}
\end{align}
which is conjugate symmetric, linear in the second component, and positive for $\langle A,A \rangle$ when $A\neq 0$. 
The properties of conjugate symmetry and linearity are straightforward to verify. The positivity property $\langle A,A\rangle\geq 0$ can be established by noting that (i) if a square matrix $M$ equals the multiplication of a matrix with its Hermitian transpose, i.e., $M=AA^\dagger$, then $M$ is a Hermitian positive semi-definite matrix; and (ii) the trace of the product of two positive semi-definite matrices $M$ and $N$ is always greater than or equal to zero, i.e., $\text{Tr} (MN)\geq 0$. 
We further restrict our consideration to the
cases where $A,B$ are not orthogonal to
$\rho$, which would yield a zero expectation
value and is thus not physically relevant for correlation functions.
$\mathcal{L}(V)$ together with the inner product of Eq.~\ref{eq:innerproduct} forms a Hilbert space.

The correlation functions are defined as complex-valued two-time functions ($\mathbb{R} \times \mathbb{R} \rightarrow \mathbb{C}$):
\begin{align}
    G_{AA}(t,t') = \mathrm{Tr} \left[ \rho A(t)^\dagger A(t') \right].
\end{align}
Given the inner product, we can show that the Green's functions
arising in the study of equilibrium and non-equilibrium
dynamics of many-body physics \cite{AGD75,Mahan,bruus,stefanucci_nonequilibrium_2013} are positive definite functions \cite{Sasvari94} of two variables.
That is, given any (finite) set of arbitrarily spaced time points $t$, the eigenvalues of the matrix that result from evaluating the Green's function at those points, $\left[G_{AA}(t_i,t_j)\right]$, are all positive or zero \footnote{Note that in the standard terminology that positive definite functions include zero eigenvalues \cite{Sasvari94}}.

We prove this by using an equivalent definition of positive definiteness, which is
that, given the aforementioned set of time points $t$,
\begin{align}
    \sum_{ij} G_{AA}(t_i,t_j) & \lambda_i^* \lambda_j \ge 0
    \label{eq:psd}
\end{align}
for any set of $\lambda_i$. This follows from
\begin{align}
\sum_{ij} &G_{AA}(t_i,t_j)  \lambda_i^* \lambda_j 
    = \sum_{i,j}  \mathrm{Tr} \left[ \rho A(t_i)^\dagger A(t_j) \right] \lambda_i^* \lambda_j \nonumber \\ \nonumber
    &= \mathrm{Tr} \left[ \rho \left( \sum_i \lambda_i A(t_i) \right)^\dagger 
    \left( \sum_j \lambda_j A(t_j) \right) \right] \\
    &= \bigg\langle  \sum_i \lambda_i A(t_i),  \sum_j \lambda_j A(t_j)\bigg \rangle \ge 0.
\end{align}
The correlation functions (and Green's function as well as
self-energies \cite{Balzer14,Gramsch15})
are therefore positive definite functions.
This holds for any operator $A \in \mathcal{L}(V)$,
and in particular it holds for fermionic
creation/annihilation operators $c^\dagger/c$, for densities $n$, for magnetization operators $n_\uparrow - n_\downarrow$, and trivially for the identity. In multi-orbital systems, the diagonal components of electronic Green's 
functions $G^{\lessgtr}_{ii}$
are positive semi-definite. 
Off-diagonal 
components can be constructed from 
linear combinations of $\langle c_i + c_j, c_i + c_j \rangle$, $\langle c_i + i c_j, c_i + i c_j \rangle$, 
and the diagonal components, all
of which are positive definite.

In the presence of time-translation invariance, ({\it i.e.} for steady-state and equilibrium systems), the 
Green's function becomes a function of a single time argument corresponding to the time difference: $G_{AA}(t-t')$,
i.e. a positive definite function of a single variable where
\begin{align}
    \sum_{ij} G_{AA}(t_i-t_j) & \lambda_i^* \lambda_j \ge 0.
    \label{eq:psd_ttinv}
\end{align}
In order for time translation to hold, the set of time points $t$  must form a group under
addition such as $\mathbb{R}$ or $\mathbb{Z}$.
With this additional structure, and as long as $\rho$ commutes
with the time evolution operator, 
\begin{align}
\langle A(t), B(t') \rangle 
=  \langle A(t-t'),  B \rangle 
= \langle A, B(t'-t) \rangle
\end{align}
follows from the cyclicity of the trace. Additional mathematical details are provided in the supplement.

The positive definiteness of the Green's function sets a strong constraint on the function and has important practical consequences, which we will explore and exploit in the remainder of the paper. In short, we will show that
(i) we can improve the signal-to-noise ratio in noisy data from experiment and theory,
and
(ii) we can construct causal extensions (with a positive spectrum) from short time data.

\begin{figure*}[!t]
    \centering
    \includegraphics[width=\textwidth]{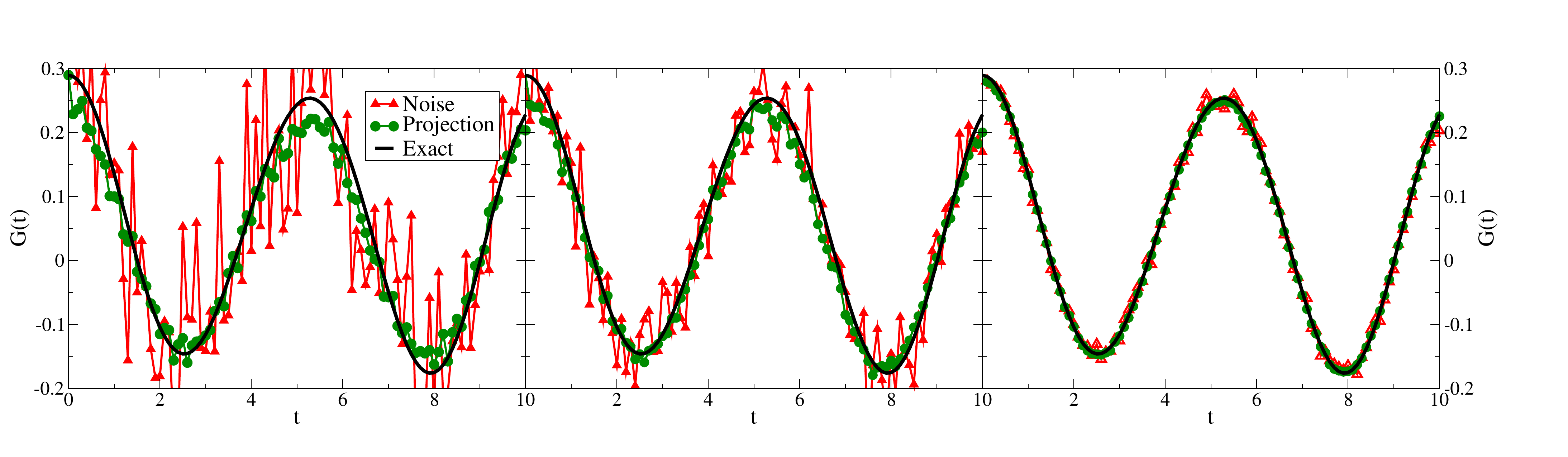}\vspace{-0.75cm}
    \caption{Real part of the on-site Hubbard dimer Green's function $G^>_{11\uparrow}(t)$. Red triangles: Data polluted by Gaussian noise with $\sigma=0.1$ (left panel), $\sigma=0.05$ (middle), and $\sigma=0.01$ (right). Green filled circles: Projection to nearest causal function. Black circles: Exact solution. Other parameters are: temperature $T = 1/10$, hopping $t = 1$, interaction $U = 5$, level energy $\epsilon=2.3$.
    }
    \label{fig:MCNoise}
\end{figure*}

\emph{Denoising correlation functions.}
A first, and natural application of the mathematics presented above, is to take a correlation function
from a source that has some inherent noise, and to use the positive definite
property to project its values to the nearest positive definite correlation function; in effect, denoising the data. Noisy
correlation functions can arise from Monte Carlo evaluations, from experimental measurements, or from simulations on
quantum computers. 

We consider a discretized, time-translation invariant correlation function on a regularly spaced time axis $t$, $G_{AA}(t_i-t_j)$.
If we label the elements of the correlation function as
$G_{AA}(t_i-t_j) \rightarrow f_{i-j}$, then
the resulting matrix $\underbar G$ with entries $\underbar G_{ij}=f_{i-j}$ is a positive semi-definite Hermitian Toeplitz matrix \cite{Sasvari94},
\begin{align}
\underbar G &=
\begin{pmatrix}
    f_0 & f_1 & f_2 & \cdots & f_n \\
    f_1^* & f_0 & f_1 &  \cdots & f_{n-1} \\
    f_2^* & f_1^* & f_0 &  \cdots & f_{n-2} \\
    \vdots & & & \ddots & \vdots \\
    f_n^* & f_{n-1}^* & f_{n-2}^* & \cdots & f_0
\end{pmatrix}.
\label{eq:gmatrix}
\end{align}
$\underbar G$ is commonly known as the Gramian or ``Gram'' matrix.

In the presence of noise in the correlation function data ($f_j$), this matrix is not positive semi-definite (PSD).
However, an alternating projection to the nearest PSD matrix \cite{Higham88} followed by projections to the nearest Toeplitz matrix \cite{Eberle03} and an enforcement of the value at time zero (which is typically known precisely, {\it e.g.} from equal-time measurements or sum rules) results in quick convergence to a positive definite function. 

The projection to the nearest PSD matrix is achieved by diagonalizing $G$ and setting all negative eigenvalues to zero. The projection to the nearest Toeplitz matrix averages entries diagonally, and the enforcement of the norm consists of fixing the diagonal to a predetermined value. While an alternating projection typically converges in less than 100 iterations, faster converging schemes, see {\it e.g.} \cite{Bauschke96,Bauschke2002PhaseRE} and \cite{Lewis09}, may be substantially more efficient; we have not explored them here. 

We illustrate the denoising in Fig.~\ref{fig:MCNoise} for
a positive definite Green's function 
of the Hubbard dimer which has been polluted with Gaussian noise.
Enforcing positive definiteness by the procedure above shows a dramatic improvement of the quality of the data. Details and a second application to state-of-the-art quantum Monte Carlo data with a continuous spectrum \cite{Erpenbeck23} are presented in the supplement.

\emph{Extending correlation functions to longer times.}
A second consequence of positive definiteness is that 
positive definite extensions of response functions to longer times exist, from which spectral measures can be constructed. We make use of two well-known mathematical facts for positive definite functions. First, functions that are only known on a subset of their domain (e.g. $\mathbb{Z}$ or $\mathbb{R}$) have
at least one extension to the full domain \cite{CaratheodoryExtension1911,krein1940problem,Sasvari2006}. This is a consequence of the extension theorems of Kre{\u \i}n (in the case of $\mathbb{R}$) and Carath\'eodory (in the case of $\mathbb{Z}$).
Second, according to Bochner's theorem \cite{Bochner1932}, the Fourier transform of a positive definite function (over $\mathbb{Z}$ or $\mathbb{R}$) is guaranteed to have positive real part; moreover, the inverse Fourier transform of a positive spectral function is guaranteed to be positive definite. 
In the discrete case, these extensions may not be unique; in the continuum case, uniqueness is guaranteed by the analyticity of the Green's function. 

Conceptually, this realization offers a straightforward methodology to obtain causal spectral functions from finite-time data: in a first step, the function is extended from short to long times by taking advantage of an extension theorem. In a second step, the uniquely defined and positive Fourier transform of the extended data is computed.

Numerically, we can use extension theorems to predict Green's function values at later times from known values at short times \cite{Sasvari94}. Supplementing a known time series $f_0, \cdots, f_n$ by a single unknown element $f_{n+1}$,
the Gramian of Eq.~\ref{eq:gmatrix} is a Toeplitz matrix with a single unknown complex number 
$f_{n+1} = \underbar G_{0,n+1} = (\underbar G_{n+1,0})^*$.
The real and imaginary parts of $f_{n+1}$ can then be obtained in a two-dimensional search in the complex plane for the region where the lowest eigenvalue of $G$ is zero or larger, i.e. $G$ is positive definite; such a value must exist \cite{CaratheodoryExtension1911} and $|f_{n+1}|\le f_0$ 
due to the positive definiteness \footnote{We thank Z. Sasv\'ari for suggesting this approach}. In practice, we find that in systems with only few well-defined excitations, and consequentially a low-rank Gram matrix, the search converges quickly to a unique solution, whereas systems with continuous spectra may present ambiguities in the extensions. 
Fig.~\ref{fig:extension} demonstrates
the approach for the Green's function of the Hubbard dimer (for details see supplement). Once the data is extended to all times, it yields a unique spectral function \cite{Bochner1932}.
\begin{figure}[tbh]
    \includegraphics[angle=270,width=\columnwidth]{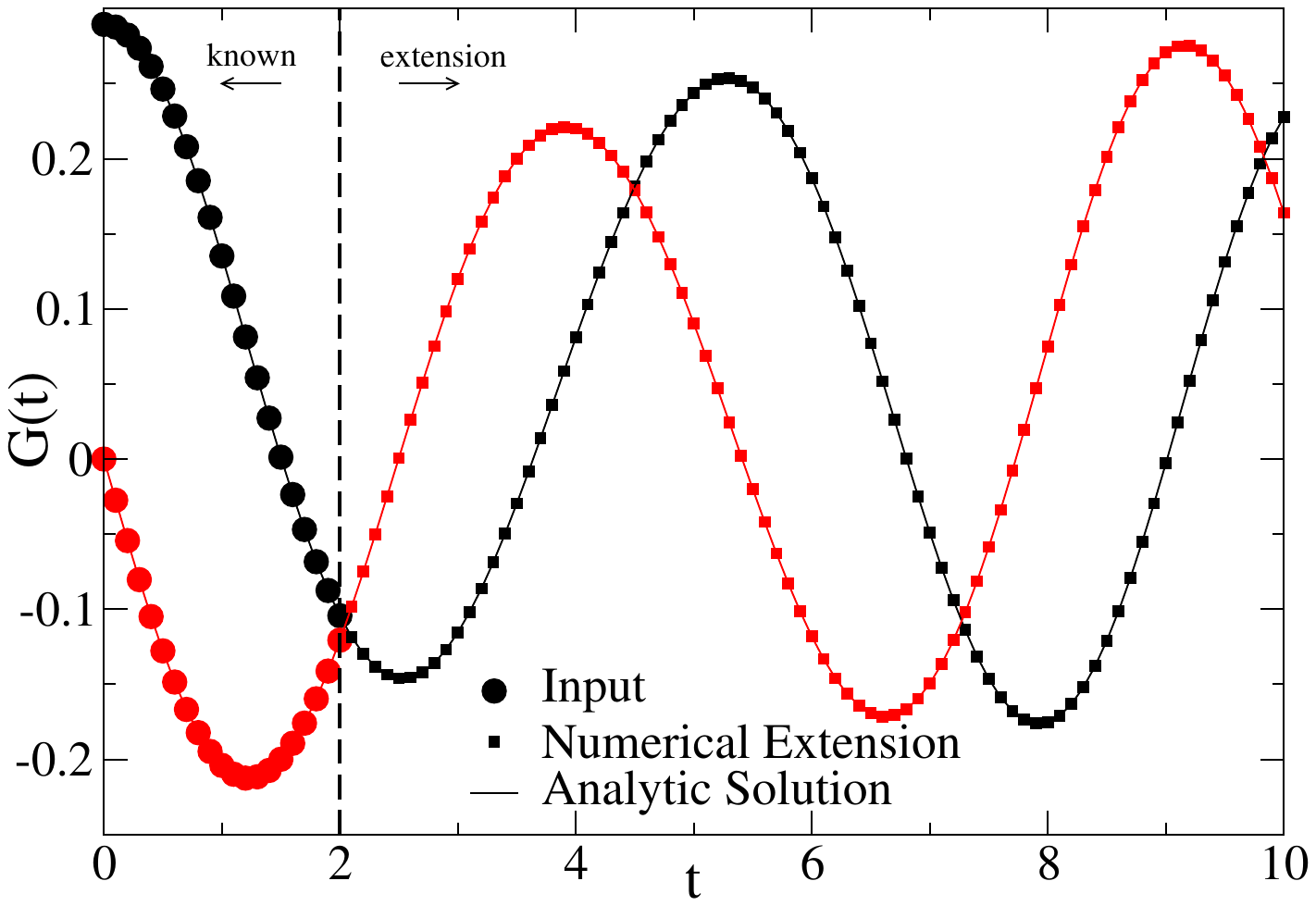}
    \caption{Extension of the on-site Hubbard dimer Green's function $G^>_{11\uparrow}(t)$. Shown are real (black) and imaginary (red) parts of input data up to $t=2$ (circles), the numerically computed extension from $t=2$ to $t=10$ (squares), and the analytically known Green's function up to $t=10$ (lines). }
    \label{fig:extension}
\end{figure}

More generally, given that the Gram matrix Eq.~\ref{eq:gmatrix} is a Toeplitz matrix of size $n$, and assuming a rank of size $r$, a classical result by Carath\'{e}odory and Fej\'{e}r \cite{CaratheodoryFejer1911} guarantees the existence of the decomposition $T=APA^\dagger$, where $A$ is a $n\times r$ Vandermonde matrix and $P$ is a $r\times r$ positive diagonal matrix. The columns of $A$ can be interpreted as uniformly sampled oscillation frequencies, and the entries of $P$ as positive `pole strengths'. This approach therefore gives both access to the spectral function in terms of a series of $r$ discrete poles at given frequencies, and a method to construct extensions for all times by enlarging the matrix $A$ with additional frequencies. The decomposition also establishes a connection to signal processing and control theory, where PSD covariance matrices of time-invariant processes are analyzed using this decomposition in super-resolution algorithms such as the multiple signal classification (MUSIC) \cite{Schmidt1986}, and where thereby approximate extensions from noisy data can be constructed. We will explore this connection, together with connections to reproducing kernel Hilbert spaces \cite{paulsen_raghupathi_2016}, in a future paper.

\emph{Application to Quantum Computing.}
As a final illustration, we apply the
denoising and extension implications of
Eq.~\ref{eq:psd_ttinv}
to a noisy correlation
function measured on
IBM's quantum computer  {\em ibm\_auckland}. 
We have measured the momentum-space
greater Green's function for the empty state
$G_k(t) := \langle 0 | c^\dagger_k(t) c_k | 0 \rangle$
for an 8-site Su-Schrieffer-Heeger
model, which is a model for free electrons with a hopping parameter that alternates with
an amplitude $\delta$
\begin{equation}
\ham = - V_{nn}\sum_i \left[1+(-\delta)^i\right] c^\dagger_i c_{i+1} + \mathrm{H.c} - \mu\sum_i c^\dagger_i c_i,
\end{equation}
where we set $V_{nn}=1$.
The raw data (originally published in Ref.~\onlinecite{kokcu2023linear}), is shown in Fig.~\ref{fig:cleaned_dataset_0.4}, in the gapped phase with $\delta=0.4$ at $k=\pi/2$. We report details of the calculation 
in the Supplementary information.
Because the model is translationally invariant and has 2 bands, the momentum basis Green's functions should exhibit at most 2 frequencies.
However, due to the hardware noise from the quantum computer, there is significant noise
in the time domain signal.  This is similarly reflected in the Fourier transform, where peaks
can be identified, but only in the power spectrum --- the spectral function (shown inFig.~\ref{fig:cleaned_dataset_0.4}) does not show the requisite analytic structure (a consistent sign across all frequencies).
We de-noise the data by asserting the positive
definiteness and performing a point-by-point optimization on the Green's function, where
the cost function is the square of the negative eigenvalues, and obtain an improved
spectrum. As can be seen in the Figure, much of the 
broad spectrum noise has disappeared and the peaks are clearly visible in the spectral function. We detail the procedure and show the
result on additional data from
different $k$ points in the supplement; we find that this iterative procedure is more suitable in the presence of extreme noise than the alternating projection scheme used for Monte Carlo data.
The final step is to extend the data as discussed above, which results in a long positive definite signal, with sharp peaks in the spectrum.

Note that applying a PSD projection as done here is 
conceptually distinct from recent efforts 
using advanced signals processing approaches 
to extract
a cleaner signal from noisy quantum simulation
\cite{cruz2022super,steckmann2023mapping,shen2023estimating}, and in fact can be used as a
complementary tool.

\begin{figure}[b]
    \centering
    \includegraphics[width=0.98\columnwidth]{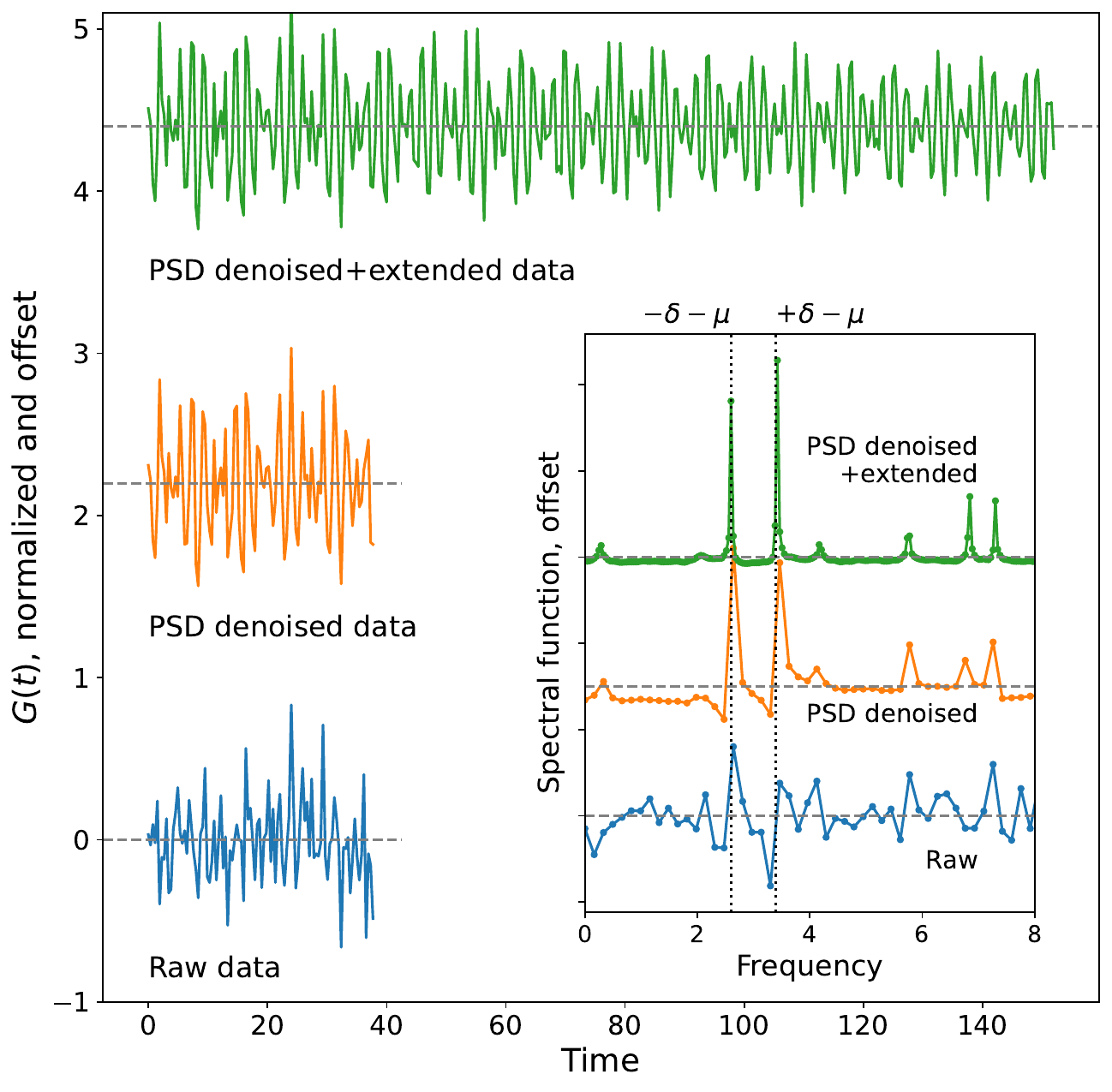}
    \caption{Main panel: greater Green's function $G_k^>(t)$ at $k=\pi/2$ of an 8-site Su-Schieffer-Heeger with $\delta=0.4$ and $\mu=-3$. The figure shows the raw data obtained from the {\em ibm\_auckland} quantum computer (blue), as well as the PSD denoised data (orange), and PSD denoised and extended data (green).
    Inset: Corresponding spectral functions. The Fourier transform used a damping factor $\tau=100$. The vertical dashed lines indicate
    the expected analytic frequencies.}
    \label{fig:cleaned_dataset_0.4}
\end{figure}

\emph{Discussion.} 
Positive semidefinite response functions are ubiquitous in many-body physics\cite{AGD75,stefanucci_nonequilibrium_2013,uhrig2019positivity}, as is the desire to reduce systematic or stochastic noise and to obtain corresponding spectral functions. The theory and algorithms presented here are broadly applicable to problems ranging from the analysis of experimental measurements to simulations of quantum systems on classical and quantum hardware.

In this Letter, we have demonstrated the efficiency of removing noise with a PSD projection in the examples of synthetic data and real-world quantum computing data, as well as the feasibility of extending data to long time. We believe that a PSD noise filter followed by an extension should be applied much more broadly to any positive definite response function, simulated or measured, before data is analyzed and/or spectra are computed, including data from quantum Monte Carlo, tensor networks, and time-resolved experiment.

While the examples in this work focused on Green's functions, a PSD projection is not limited
to this particular use case
since the operators 
in Eq.~\ref{eq:psd}
are general.
For example, 
the identity operator is included,
and thus quantities of the form
$\langle \psi_0 | e^{-i \ham t} | \psi_0 \rangle$ are also amenable to denoising
in the same way.  This covers a range
of novel quantum algorithms that make use of such quantities,
including subspace expansions \cite{cortes2022fast,klymko2022real} and
energy minimization based on quantum phase estimation \cite{ding2023even,*ding2023simultaneous},
as well as Loschmidt echo measurements \cite{schuckert2022probing}.

We note that, similar to the extension techniques presented in this letter, memory kernel approaches \cite{Shi03,Cohen11,Cohen13,Cerrillo14,Pollock18,Pollock22} pioneered in transport physics  aim to extend data measured for short times to longer times. The approach presented is complementary to those techniques and may be combined with them. 

The connection highlighted here between positive definite response functions and the Toeplitz Vandermonde decomposition of the Gram matrix in addition allows to make a connection to algorithms developed for signal processing and control theory, such as MUSIC \cite{Schmidt1986}, which are used to analyze positive semidefinite covariance matrices. This analysis will be the topic of a future paper.

\vspace{0.1in}
\begin{acknowledgments}
We thank Vadym Adamyan for making us aware of extension theorems, Zoltan Sasv\'ari for pointing out the possibility of considering a Gram matrix extended by a single point, Efekan K\"okcu for
suggestions regarding non-equilibrium
Green's functions,
and Michael Ragone for feedback on the
mathematical exposition.
AFK was supported by the Department of Energy, Office of Basic Energy Sciences, Division of Materials Sciences and Engineering under grant no. DE-SC0023231.  
EG and CY were supported by the  U.S. Department of Energy, Office of Science, Office of Advanced Scientific Computing Research and Office of Basic Energy Science, Scientific Discovery through Advanced Computing (SciDAC) program under Award Number(s) DE-SC0022088.
The data for the figures are submitted as a supplement to this paper.
\end{acknowledgments}

\bibliography{refs}

\clearpage
\onecolumngrid
\appendix

\renewcommand\thefigure{S\arabic{figure}}  
\renewcommand\thetable{S\arabic{table}}  
\setcounter{figure}{0}

\section{Application to quantum computer data}

Fig.~\ref{fig:QCdata} shows the greater Green's function obtained 
from IBM quantum computers for the Su-Schrieffer-Heeger (SSH) model
\begin{align}
    \ham = -\sum_{\langle i,j\rangle } \left[ V_{nn} + (-1)^i \delta/2 \right] c^\dagger_i c_j  - \mu\sum_i c^\dagger_i c_i.
    \label{eq:SSHmodel_SI}
\end{align}
This data was originally presented in K\"okc\"u et al.~\cite{kokcu2023linear}.
The figure shows the data for $\delta=0.4$, $\mu=-3$,
and momenta $k \in \lbrace 0,1,2,3,4 \rbrace \frac{2\pi}{8}$.
Note that while the Fourier transform power spectrum shows peaks near the expected values (indicated as vertical dashed lines), there is a substantial amount of noise present.

\begin{figure}[htpb]
    \centering
    \includegraphics[width=0.6\textwidth]{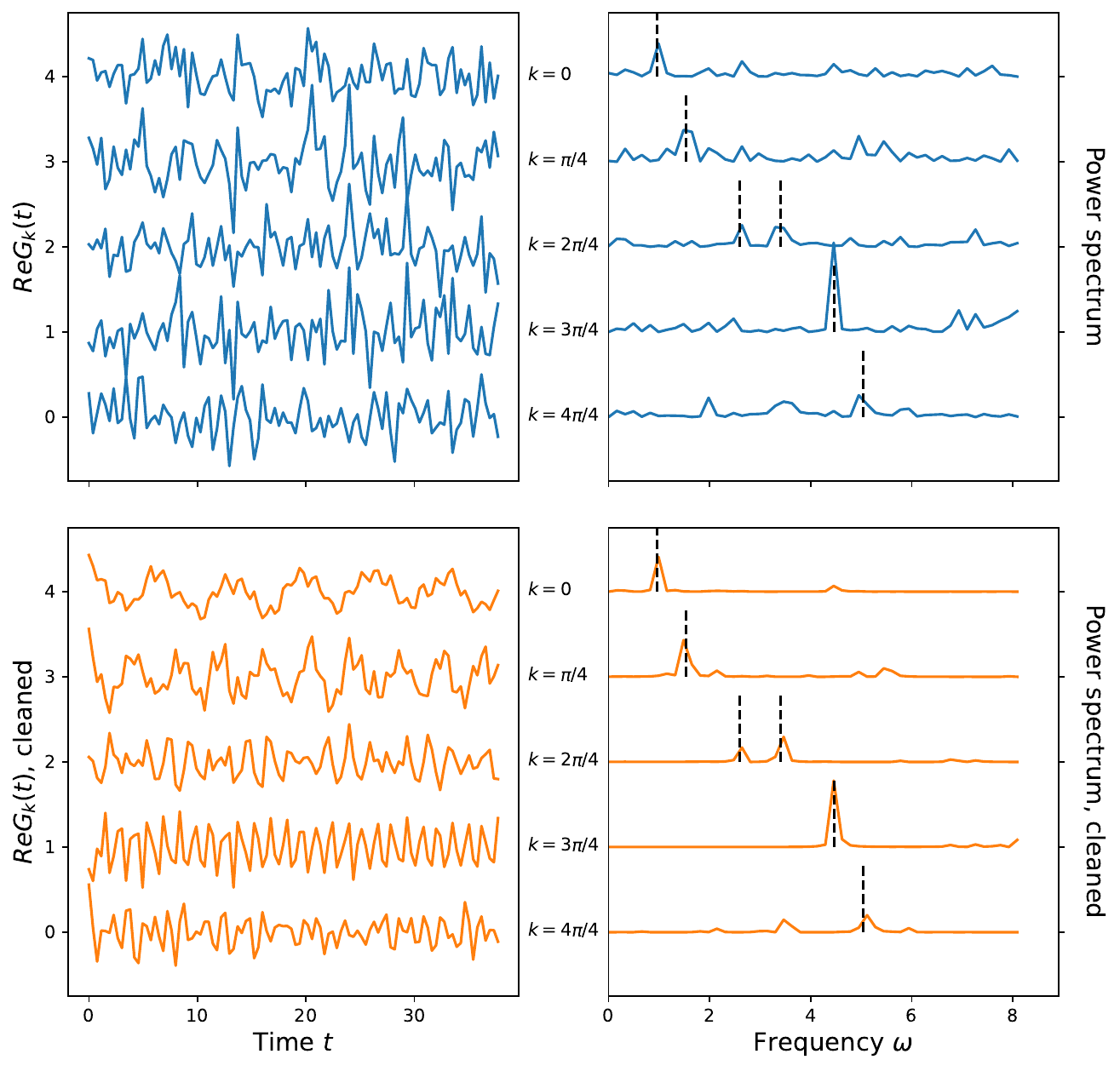}
    \caption{Top panel: Raw data for an 8-site Green's function 
    $\mathrm{Re} G_k(t)$ and corresponding power spectrum
    for the SSH model with $\delta=0.4$ and $\mu=-3$ obtained from IBM's quantum computer {\em ibm\_auckland} (originally from Ref.~\cite{kokcu2023linear}). Bottom panel:
    Data after PSD denoising as discussed below.
    Individual curves are $k$ points. The dashed lines indicate the expected frequencies obtained by analytically solving Eq.~\ref{eq:SSHmodel_SI}.}
    \label{fig:QCdata}
\end{figure}

The Gramian of the Green's function, {\it i.e.} the Toeplitz matrix constructed from the Green's function, must follow the structure of \cref{eq:gmatrix} and be positive semidefinite. In the presence of noise, the Gramian is a Toeplitz matrix but is not positive definite.
Defining a cost function
\begin{align}
    \mathcal{C} = \sum_i \left[ \lambda_i \left(\mathrm{sgn}(\lambda_i) - 1\right) \right]^2
\end{align}
where $\lambda_i$ are the eigenvalues of $G$, we cycle through
the individual entries in $G$ and numerically optimize them with respect to this
cost function. In cases with very large noise, such as with data from quantum computers, we find this procedure to be more reliable than the alternating projection described in the main text.

We apply this approach to the data set from Ref.~\cite{kokcu2023linear}
with the worst
quality, $\delta = 0.4$; the results are shown in 
\cref{fig:QCdata}. For most of the momenta,
the assertion that the Green's function should have the
structure of \cref{eq:gmatrix} yield
quite good results --- the algorithm correctly picks out the
correct peak(s) in the Fourier transform.  The notable exception
is the $k=\pi$, which is because the original data is particularly
poor.  



\section{Application to steady-state Monte Carlo data}

\begin{figure}[htpb]
    \centering
    \includegraphics[angle=270,width=0.49\textwidth]{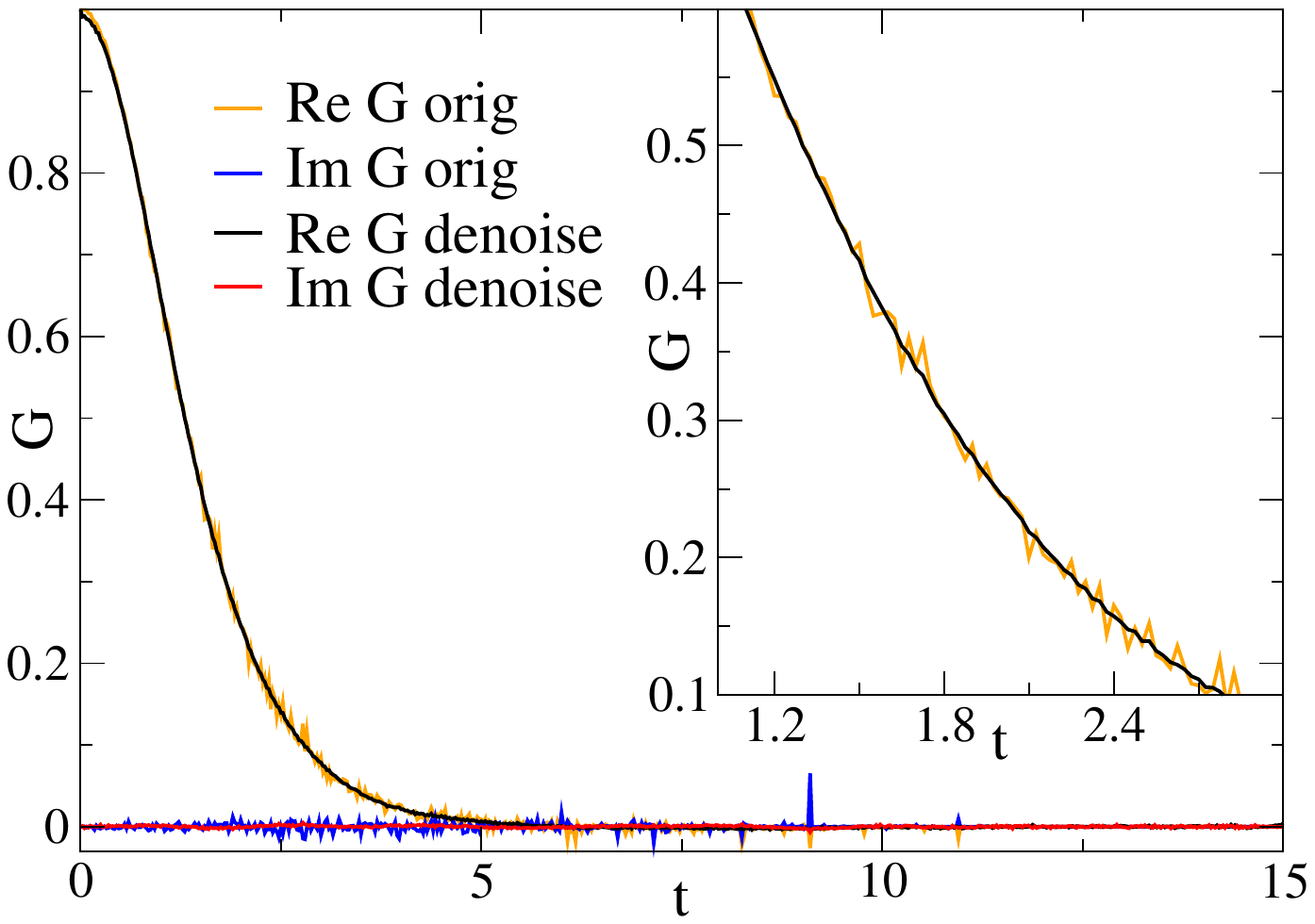}
    \caption{De-noised Quantum Monte Carlo data. Main panel: inchworm \cite{Cohen15} quantum Monte Carlo of a steady-state transport problem (for data and parameters see Fig.~1 of \cite{Erpenbeck23}), real part (orange) and imaginary part (blue). Denoised data, real part (black) and imaginary part(red). Inset: Zoom into the region near t=2 highlighting the effect of denoising.}
    \label{fig:steady_state_qmc}
\end{figure}

In this section, we present de-noised data from a real-time simulation with current state-of-the-art quantum Monte Carlo technology. We use the steady-state formulation \cite{Erpenbeck23} of the Inchworm \cite{Cohen15} Monte Carlo real-time propagation scheme. The method overcomes the dynamical sign problem inherent to many real-time Monte Carlo methods \cite{Gull11RMP,Muhlbacher08} and is therefore able to simulate transport data for long times.

The data presented has Monte Carlo noise whose size differs as a function of time but is largest in the region near $t=2$. Visible are also data `glitches' near $t=9.1$ which are likely due to ergodicity issues in the simulation and which are clearly unphysical.

We employ the alternate projection described in the main text, while keeping the norm fixed to the known analytical value of 1. It is evident that the causal projection onto a positive definite function eliminates much of the Monte Carlo noise and results in a real-time function that is much smoother. Unphysical outliers in the data are eliminated. Obtaining data of the same quality as the de-noised data would require a substantially larger investment of computer time.

\section{The Hubbard Dimer -- Details}
For the data in Fig.~\ref{fig:MCNoise}, we consider a Hubbard dimer (a two-site Hubbard model) with the following Hamiltonian:
\begin{align}
H &= H_0 + H_I,\\
    H_0 &= -\epsilon\sum_{i=1,2}\left(n_{i\uparrow}+n_{i\downarrow}\right) -v \sum_{\sigma=\uparrow,\downarrow} c_{0\sigma}^\dagger c_{1\sigma} + c_{1\sigma}^\dagger c_{0\sigma},\\
    H_I &= U \sum_{i=1,2} \left(n_{i\uparrow}n_{i\downarrow}-\frac{1}{2}(n_{i\uparrow}+n_{i\downarrow})\right)
\end{align}
The operators $c_{i\sigma}$ describe four fermion annihilation operators that annihilate particles on site $i=1,2$ with spin $\sigma=\uparrow,\downarrow$. Their adjoints $c_{i\sigma}^dagger$ create the corresponding particles on site $i$ with spin $\sigma$. These operators correspond to matrices of size  $16\times16$ in Fock space and satisfy the fermion commutation rules.

The parameter $U$ describes the on-site repulsion and is set to $U=5$. $\epsilon$ is an on-site level energy term ($\epsilon=0$ corresponds to half filling), and $v$ describes the hopping between sites $1$ and $2$. At temperature $T=1/\beta=0.1,$ the choice of $\epsilon=2.3$ corresponds to a density of $0.710556$ on all spin-orbitals.

The Green's function discussed in Fig.~\ref{fig:MCNoise} is defined as
\begin{align}
    G_{11\uparrow}(t)=\frac{1}{Z}\text{Tr}\left[e^{(-\beta+it)H} c_{1\uparrow} e^{-itH} c_{1\uparrow}^\dagger\right],
\end{align}
and is a positive definite function that is symmetric in the real part and anti-symmetric in the imaginary part, with $G_{11\uparrow}(t=0)=0.289444.$ (note that the system is degenerate in site and spin indices). The positive definite function corresponds, up to a factor of $-i$, to the spin-up on-site greater Green's function of site $1$ of the system.

\section{Additional Mathematical notes}
In an attempt to make the presentation in the main text accessible to most physicists, we have eliminated some of the more formal mathematical aspects. We would like to highlight the following additional connections:

\subsection{Inner products induce positive definite functions}
We reframe the time dependence of the
operators as a map on $\mathcal{L}(V)$.
Note that the time argument
is an element of the group of real numbers $\mathbb{R}$ in the continuous case. When time is discretized  (e.g. due to a numerical sampling on an equidistant time grid), it is indexed by an element the group
of integers $\mathbb{Z}$; we will generically
use $\mathcal{G}$ to refer to the group. 
The time evolution of the operators in the Heisenberg picture
is a map $\phi_t$ that acts on elements of $\mathcal{L}(V)$. That is, for an operator $A(t') \in \mathcal{L}(V)$ and a time argument $t \in \mathbb{R}$,
\begin{align}
\phi_t :\ &\mathcal{L}(V) \rightarrow \mathcal{L}(V) \\
& \phi_t A(t') = A(t+t')
\end{align}
It is straightforward to show that $\phi_t$ preserves that group structure of the
time argument: $\phi_t\phi_{t'} = \phi_{t+t'}$.  Thus, it is a representation of $\mathcal{G}$.
Moreover, for the inner product in Eq.~\ref{eq:innerproduct}, this map is unitary.
This can be seen from the relations
\begin{align}
G_{AB}(t-t') =  \langle A(t), B(t') \rangle 
=  \langle \phi_t A, \phi_{t'} B \rangle \\
= \langle \phi_{t-t'} A, B \rangle
= \langle A, \phi_{t'-t} B \rangle
\end{align}
which follow from the cyclicity of the trace.
Combining these two points, and because $\phi$ is a unitary representation
of $\mathcal{G}$, we can show the following theorem.
\begin{theorem}
\label{psd-thm}
Consider a correlation function of two operators on
Fock space $\langle A(t_i), A(t_j) \rangle$.
This correlation function is a positive definite function.
\end{theorem}
\begin{proof}
For a positive definite function evaluated on an arbitrary set of time points $t$, $f(t_i,t_j)$, the relation 
\begin{align}
    \sum_{i,j}f(t_i,t_j) \lambda_i^* \lambda_j \geq 0
\end{align}
holds for any complex vector $\lambda$. Expanding out the
definition of $\langle A(t_i), A(t_j) \rangle$, it
is straightforward to show that
\begin{align}
    \langle A(t_i), A(t_j) \rangle = \mathrm{Tr}\left[\rho\left(
    \sum_i \lambda_i \phi_i(A)\right)^\dagger \left( \sum_j \lambda_j \phi_j(A)\right)
    \right],
\end{align}
which is the inner product of an element of $\mathcal{L}(V)$ with itself and thus must be $\ge 0$.
\end{proof}
Thus,
$G_{AA}(t-t')$ is a positive definite function.

\subsection{Connection to Reproducing Kernel Hilbert Spaces}
Inner products induce reproducing kernel Hilbert spaces. Given a Hilbert space $\mathcal{L}$ with inner product $\langle \cdot ,\cdot \rangle$, the function $K(x,y)=\langle x, y \rangle$ is a kernel function and determines a corresponding reproducing kernel Hilbert space \cite{paulsen_raghupathi_2016}
on the space of functions  $\mathbb{R} \rightarrow \mathbb{C}$.
Reproducing kernel Hilbert spaces form an important tool in statistics and machine learning, with connections to interpolation and approximation theory. This connection will be further examined in a follow-up paper.
\end{document}